\documentclass[11pt]{article}
\usepackage[margin=1in]{geometry}
\usepackage{amsmath,amssymb,amsthm,mathtools}
\usepackage{booktabs}
\usepackage{enumitem}
\usepackage{array}
\usepackage{algorithm}
\usepackage{algpseudocode}
\usepackage{placeins}
\usepackage{hyperref}
\hypersetup{hidelinks}
\usepackage{microtype}

\newcolumntype{L}[1]{>{\raggedright\arraybackslash}p{#1}}

\newtheorem{definition}{Definition}

\newtheorem{theorem}{Theorem}
\newtheorem{corollary}{Corollary}
\newtheorem{proposition}{Proposition}

\newcommand{\Rset}{\mathcal R}
\newcommand{\Eset}{\mathcal E}
\newcommand{\Fadm}{\mathcal F_{\mathrm{adm}}}
\newcommand{\Pe}{\mathcal P_e}
\newcommand{\Oset}{\mathcal O}
\newcommand{\Aset}{\mathcal A}
\newcommand{\Cset}{\mathsf C}
\newcommand{\Rls}{\mathsf R}
\newcommand{\Dset}{\mathsf D}
\newcommand{\Lset}{\mathsf L}
\newcommand{\Block}{\mathsf{block}}
\newcommand{\supp}{\operatorname{supp}}
\newcommand{\partic}{\operatorname{part}}
\newcommand{\resmag}{\mu}
\newcommand{\kappares}{\kappa}
\newcommand{\certadm}{\operatorname{adm}_{\mathrm{cert}}}
\newcommand{\resadm}{\operatorname{adm}_{\mathrm{res}}}
\newcommand{\starop}{\operatorname{St}}
\newcommand{\labels}{\operatorname{Lab}}
\newcommand{\plabels}{\operatorname{PLab}}

\title{Quotient Admission Algorithms for Witness-Supported Graph Windows}
\author{Yushan Li}
\date{}

\begin{document}
\maketitle

\begin{abstract}
We formulate the quotient admission problem for finite graph-window rows. The input is a finite row set, an admissible evidence map, semantic labels, witness-support hypergraphs, and atom-level admissibility predicates. The output is a quotient decision on evidence atoms, with possible decisions certificate, residual, low-confidence, or blocked. The problem asks for the maximal guard-respecting atom-level decision map that uses no refinement beyond the admissible evidence partition. We prove an atom-union characterization of identifiable classes, give a witness-support hypergraph guard for certificate admission, characterize projected-label conflicts as blocked atoms, and present quotient admission algorithms with correctness, maximality, and complexity guarantees. With explicit evidence vectors and hyperedges, the algorithms run in expected $O(B+I+n)$ time and space by hashing and deterministic $O(B+I+n\log n)$ time by sorting under a key-linear comparison model, where $n$ is the number of rows, $B$ is the total evidence encoding length, and $I$ is the total hyperedge incidence size. We also prove a magnitude-only indistinguishability lower bound: any evaluator that observes only residual magnitudes fails on instances whose evidence atoms require different residual decisions after the magnitudes collapse them.
\end{abstract}

\section{Introduction}
Many finite decision procedures first construct local rows and then decide what can be exposed from the information carried by each row. The row may come from a graph window, a local consistency check, a finite hypothesis family, or a candidate-and-witness construction. Once the row is fixed, the algorithmic object is a quotient: which decision is identifiable from the admissible evidence, and which row classes must be blocked because the evidence partition cannot separate their projected decisions?

This paper studies that question as a finite combinatorial problem. A finite row set $\Rset$ is mapped by an admissible evidence map $e:\Rset\to\Eset$ into an evidence space. The fibers of $e$ form an atom partition $\Pe$. Any decision rule that uses only admissible evidence must be constant on these atoms and therefore factors through the quotient map $q_e:\Rset\to\Pe$. The remaining task is to define a feasible and maximal atom-level admission map
\[
F:\Pe\to\{\Cset,\Rls,\Lset,\Block\},
\]
where $\Cset$ denotes certificate admission, $\Rls$ denotes residual admission, $\Lset$ denotes low-confidence admission, and $\Block$ denotes a blocked atom.

The witness structure is represented by finite hypergraphs. For each row $r$, a hypergraph $H_r=(V_r,\mathcal W_r)$ contains the candidate token $q_r$ and witness-support tokens. The support surface is the union of all witness hyperedges, and candidate participation records whether some witness hyperedge contains $q_r$. Certificate admission requires both a certificate-admissibility predicate and this hypergraph guard. Residual admission is governed by a separate residual-admissibility predicate and a support-compatibility predicate $\kappa$. This separation keeps residual magnitude, support compatibility, and admission feasibility as distinct finite predicates.

Several application areas instantiate this row model. In graph optimization and factor-graph reasoning, local windows and loop candidates often carry witness and residual information~\cite{kuemmerle2011g2o,grisetti2010tutorial,dellaert2017factorgraphs,mangelson2018pcm,yang2024groupk,tian2016multimodel,burnett2023tbv}. In selective prediction, reject-option classification, and learning to defer, an algorithm exposes a prediction only on selected inputs~\cite{chow1970reject,elyaniv2010selective,cortes2016rejection,geifman2019selectivenet,madras2018defer}. The present paper abstracts the common finite quotient structure: evidence atoms, witness-support set systems, blocking of indistinguishable projected conflicts, and conservative admission.

\paragraph{Contributions.}
The contributions are as follows.
\begin{enumerate}[leftmargin=2em]
    \item We define the quotient admission problem on a finite row space with admissible evidence atoms, semantic labels $\Oset=\{\Cset,\Rls,\Dset,\Lset\}$, and admission decisions $\Aset=\{\Cset,\Rls,\Lset,\Block\}$.
    \item We prove that identifiable classes are exactly unions of evidence atoms. Consequently, the identifiable classes form a Boolean algebra isomorphic to $2^{\Pe}$ and there are $2^{|\Pe|}$ identifiable classes.
    \item We model witness support as a finite hypergraph and show, under a certificate-validity condition, that certificate admission requires a support guard: nonempty support surface, candidate participation, and a certificate-admissibility predicate.
    \item We give \textsc{BuildAtoms} and \textsc{QuotientAdmission}, prove correctness and maximality under the no-hidden-refinement rule, and prove expected linear and deterministic near-linear complexity bounds under explicit encodings.
    \item We prove a magnitude-only indistinguishability lower bound showing that residual magnitudes cannot identify residual admission after they collapse atoms with different support compatibility.
\end{enumerate}

\paragraph{Organization.}
Section~\ref{sec:related} places the problem among partition refinement, finite quotients, hypergraphs, CSP-style filtering, and application-facing admission problems. Section~\ref{sec:problem} states the quotient admission problem. Sections~\ref{sec:atoms} and~\ref{sec:support} give the atom and hypergraph structure. Section~\ref{sec:algorithms} presents the algorithms, correctness and maximality theorem, and complexity bound. Section~\ref{sec:lower-bound} gives the magnitude-only lower bound. Section~\ref{sec:example} gives a finite example.

\section{Related Algorithmic and Combinatorial Problems}\label{sec:related}
Partition refinement and finite quotient construction are the closest algorithmic relatives. Classical partition-refinement algorithms group finite states into equivalence classes that cannot be distinguished by the available observations or transitions~\cite{hopcroft1971nlogn,paige1987partition}. Here the partition is induced directly by an evidence map, and the algorithm computes the admissible quotient decision on the resulting atoms.

Hypergraphs and set systems provide the witness-support language~\cite{berge1989hypergraphs}. In quotient admission, hyperedges define support surfaces and candidate participation predicates that act as guards on atom-level decisions.

Finite-domain constraint satisfaction and partial constraint satisfaction give a useful admission and filtering vocabulary~\cite{freuder1992partial,brailsford1999csp}. In those settings one filters or relaxes assignments. In quotient admission, the assignment space is replaced by evidence atoms, and feasibility is determined by whether the desired decision can be made without refining the evidence partition.

Structural consistency selection methods reason over compatible measurement sets and higher-order consistency structures~\cite{mangelson2018pcm,yang2024groupk}. Their objective is to select consistent sets. The quotient admission problem instead assumes a finite row representation and asks which atom-level routing labels are identifiable from that representation.

Reject-option classification, selective prediction, and learning to defer introduce output decisions that abstain from ordinary prediction~\cite{chow1970reject,elyaniv2010selective,cortes2016rejection,geifman2019selectivenet,madras2018defer}. Quotient admission replaces the statistical risk objective in those settings by a finite evidence partition and hypergraph guards.

Graph optimization, factor graphs, loop verification, synchronization, and certifiable estimation motivate row constructions in which witness and residual information coexist~\cite{kuemmerle2011g2o,grisetti2010tutorial,dellaert2017factorgraphs,tian2016multimodel,burnett2023tbv,singer2011angular,rosen2019sesync,yang2021teaser}. These domains supply examples of finite rows with graph-derived evidence; the results below analyze the induced quotient problem.

\section{Problem Statement}\label{sec:problem}
We first state the finite quotient admission problem in full. All sets in this section are finite.

\begin{definition}[Graph-window row]
A graph-window row is a tuple
\[
r=(G_r,q_r,W_r,\eta_r),
\]
where $G_r$ is a finite graph window, $q_r$ is a distinguished candidate token, $W_r$ is a finite witness pool, and $\eta_r$ is any remaining admissible metadata.
\end{definition}

\begin{definition}[Admissible evidence and atoms]
Let $e:\Rset\to\Eset$ be an admissible evidence map. The evidence sigma-algebra is
\[
\Fadm=\sigma(e).
\]
Define an equivalence relation $\sim_e$ on $\Rset$ by
\[
r\sim_e r' \quad\Longleftrightarrow\quad e(r)=e(r').
\]
For every $v\in\operatorname{im}(e)$, define the atom
\[
A_v=e^{-1}(v).
\]
The atom partition $\Pe=\{A_v:v\in\operatorname{im}(e)\}$ is the set of $\sim_e$-equivalence classes, and the quotient map $q_e:\Rset\to\Pe$ sends a row to its evidence atom.
\end{definition}

\begin{definition}[Semantic labels and admission decisions]
The semantic label set is
\[
\Oset=\{\Cset,\Rls,\Dset,\Lset\},
\]
where $\Cset$ denotes certificate-primary, $\Rls$ denotes residual-primary, $\Dset$ denotes dual information, and $\Lset$ denotes low-confidence. The admission output set is
\[
\Aset=\{\Cset,\Rls,\Lset,\Block\}.
\]
The conservative projection $\pi:\Oset\to\Aset$ is
\[
\pi(\Cset)=\Cset,\qquad
\pi(\Rls)=\Rls,\qquad
\pi(\Dset)=\Lset,\qquad
\pi(\Lset)=\Lset.
\]
\end{definition}

\begin{definition}[Quotient admission input]
An instance of the quotient admission problem consists of
\[
(\Rset,e,\ell,\{H_r\}_{r\in\Rset},\supp,\partic,\certadm,\kappares,\resadm),
\]
where $\ell:\Rset\to\Oset$ is a semantic labeling, $H_r$ is a witness-support hypergraph for row $r$, and
\[
\supp,\partic,\certadm,\kappares,\resadm:\Rset\to\{0,1\}
\]
are $\Fadm$-measurable predicates. The predicates respectively encode nonempty support, candidate participation, certificate admissibility, residual support compatibility, and residual admissibility.
\end{definition}

\begin{definition}[Feasible quotient admission map]
A feasible quotient admission map is a function
\[
F:\Pe\to\Aset
\]
satisfying the following constraints for every atom $A\in\Pe$.
\begin{enumerate}[label=(F\arabic*),leftmargin=2.4em]
    \item \textbf{Atom-level output.} The decision depends only on $A$, not on a row inside $A$.
    \item \textbf{Projected-label blocking.} $F(A)=\Block$ if and only if $\pi\circ\ell$ is not constant on $A$.
    \item \textbf{Certificate guard.} If $F(A)=\Cset$, then $(\pi\circ\ell)(A)=\Cset$, $\certadm(A)=1$, $\supp(A)=1$, and $\partic(A)=1$.
    \item \textbf{Residual guard.} If $F(A)=\Rls$, then $(\pi\circ\ell)(A)=\Rls$, $\resadm(A)=1$, and $\kappares(A)=1$.
    \item \textbf{Conservative projection.} If $(\pi\circ\ell)(A)=\Lset$ and $\pi\circ\ell$ is constant on $A$, then $F(A)=\Lset$.
    \item \textbf{No hidden refinement.} No branch of $F$ may use a predicate that is not $\Fadm$-measurable.
\end{enumerate}
Here $(\pi\circ\ell)(A)$ and the predicate values on $A$ are well defined whenever the relevant function is constant on $A$.
\end{definition}

\begin{definition}[Maximal feasible map]
A feasible quotient admission map is maximal if, for every atom $A$ on which $\pi\circ\ell$ is constant, it outputs $\Cset$ whenever the projected label and certificate guard allow $\Cset$, and it outputs $\Rls$ whenever the projected label and residual guard allow $\Rls$. All remaining nonblocked atoms are mapped to $\Lset$.
\end{definition}

\section{Evidence Atoms and Boolean Structure}\label{sec:atoms}
The first structural fact is the standard finite quotient characterization. It is stated here because every later admissibility constraint is imposed on atoms.

\begin{theorem}[Atom-union characterization]\label{thm:atom-union}
For $P\subseteq\Rset$, the following are equivalent:
\begin{enumerate}[label=(\roman*),leftmargin=2em]
    \item $P$ is $\Fadm$-measurable.
    \item $P$ is a union of atoms in $\Pe$.
    \item There exists $S\subseteq\operatorname{im}(e)$ such that $P=e^{-1}(S)$.
\end{enumerate}
Moreover, every map $f:\Rset\to X$ into a finite set $X$ that is constant on atoms factors uniquely as $f=\bar f\circ q_e$ for a map $\bar f:\Pe\to X$.
\end{theorem}

\begin{proof}
Since $\Rset$ is finite, $\sigma(e)$ is the finite Boolean algebra generated by the fibers of $e$. Thus every $\Fadm$-measurable set is a union of fibers, and the fibers are exactly the atoms in $\Pe$. Conversely, any union of atoms is the inverse image under $e$ of the corresponding subset of $\operatorname{im}(e)$. This proves the equivalence.

For factorization, define $\bar f(A)=f(r)$ for any $r\in A$. The definition is independent of the representative because $f$ is constant on atoms. Since $q_e$ is surjective onto $\Pe$, the factorization is unique.
\end{proof}

\begin{corollary}[Boolean algebra of identifiable classes]\label{cor:boolean}
The identifiable subclasses of $\Rset$ form a finite Boolean algebra isomorphic to $2^{\Pe}$.
\end{corollary}

\begin{proof}
By Theorem~\ref{thm:atom-union}, every identifiable subclass is uniquely determined by the set of atoms that it contains. The Boolean operations correspond to union, intersection, and complement of these atom index sets.
\end{proof}

\begin{corollary}[Counting identifiable classes]\label{cor:count}
The number of identifiable classes is $2^{|\Pe|}$.
\end{corollary}

\begin{proof}
Each identifiable class is uniquely specified by choosing an arbitrary subset of the atom set $\Pe$.
\end{proof}

\begin{corollary}[Exact-twin indistinguishability]\label{cor:twins}
If $r,r'\in\Rset$ satisfy $e(r)=e(r')$, then every admissible evaluator assigns them the same quotient decision. Hence rows whose projected semantic labels differ inside one atom cannot be separated without violating the no-hidden-refinement rule.
\end{corollary}

\begin{proof}
The equality $e(r)=e(r')$ places $r$ and $r'$ in the same atom. Every admissible evaluator factors through $q_e$ and is therefore constant on that atom.
\end{proof}

\section{Witness-Support Hypergraphs}\label{sec:support}
The witness component is a finite set system. It supplies guards for certificate admission but does not refine the evidence partition unless the derived predicates are included in the admissible evidence map.

\begin{definition}[Witness-support hypergraph]
For row $r$, let $H_r=(V_r,\mathcal W_r)$ be a finite hypergraph with $q_r\in V_r$. The vertices encode the candidate token and witness-support tokens. Each hyperedge $W\in\mathcal W_r$ satisfies $W\subseteq V_r$, so membership statements such as $q_r\in W$ are well defined. Define
\[
\Sigma_r^{\mathrm{sup}}=\bigcup_{W\in\mathcal W_r} W,\qquad
\starop_r(q_r)=\bigcup\{W\in\mathcal W_r:q_r\in W\}.
\]
The set $\Sigma_r^{\mathrm{sup}}$ is the support surface and $\starop_r(q_r)$ is the candidate-support surface. The row predicates are
\[
\supp(r)=\mathbf 1[\Sigma_r^{\mathrm{sup}}\neq\varnothing],
\qquad
\partic(r)=\mathbf 1[\starop_r(q_r)\neq\varnothing].
\]
\end{definition}

\begin{definition}[Atom-level guards]
The predicates $\certadm$, $\kappares$, and $\resadm$ are atom-level, $\Fadm$-measurable predicates. The same is required of $\supp$ and $\partic$ when they are used by a quotient admission map. Thus each predicate has a well-defined value on every atom $A\in\Pe$.
\end{definition}

\begin{definition}[Certificate-valid atom]
An atom $A$ is certificate-valid if some row $r\in A$ satisfies $\Sigma_r^{\mathrm{sup}}\neq\varnothing$, there exists $W\in\mathcal W_r$ with $q_r\in W$, and $\certadm(A)=1$. Because the predicates are $\Fadm$-measurable, these conditions then hold throughout the atom.
\end{definition}

\begin{proposition}[Certificate-validity implies the support guard]\label{prop:support-guard}
Let $F:\Pe\to\Aset$ be an atom-level decision map that uses no hidden refinement and assigns $F(A)=\Cset$ only when $A$ is certificate-valid. Then $F(A)=\Cset$ implies $\certadm(A)=1$, $\supp(A)=1$, and $\partic(A)=1$.
\end{proposition}

\begin{proof}
Fix an atom $A$ with $F(A)=\Cset$. By the hypothesis on $F$, the atom $A$ is certificate-valid. Hence $A$ contains a row $r$ with $\Sigma_r^{\mathrm{sup}}\neq\varnothing$, so $\supp(r)=1$. The same row has a witness hyperedge $W\in\mathcal W_r$ with $q_r\in W$, so $\starop_r(q_r)\neq\varnothing$ and $\partic(r)=1$. Certificate validity also gives $\certadm(A)=1$. Since $\supp$ and $\partic$ are $\Fadm$-measurable, they are constant on $A$. Therefore $\certadm(A)=1$, $\supp(A)=1$, and $\partic(A)=1$.
\end{proof}

\begin{proposition}[Residual admissibility is not residual magnitude]\label{prop:residual-gate}
Residual magnitude can be used only after residual admissibility and support compatibility have been established at the atom level. In particular, a feasible quotient admission map cannot replace $\resadm$ and $\kappares$ by a scalar residual magnitude unless that scalar is sufficient to recover those predicates on atoms.
\end{proposition}

\begin{proof}
The residual guard in the problem definition requires $\resadm(A)=1$ and $\kappares(A)=1$ before residual admission is possible. A scalar magnitude is not one of these predicates. It may be part of the admissible evidence vector, but it determines residual admission only if the corresponding atom-level predicates are functions of that scalar on the instance.
\end{proof}

\section{Algorithms and Correctness}\label{sec:algorithms}
The algorithms are stated on explicitly encoded finite inputs. \textsc{BuildAtoms} constructs the evidence atoms and verifies that guard predicates used by the admission rule are atom-level. \textsc{QuotientAdmission} applies the maximal feasible decision rule.

\begin{algorithm}[t]
\caption{\textsc{BuildAtoms}}
\label{alg:build-atoms}
\begin{algorithmic}[1]
\Require finite rows $\Rset$, evidence map $e$, labels $\ell$, hypergraphs $\{H_r\}$, predicates $\certadm,\kappares,\resadm$
\Ensure atom records or a report that a predicate is not $\Fadm$-measurable
\State initialize an empty dictionary $D$
\For{each row $r\in\Rset$}
    \State $v\gets e(r)$
    \State compute $\supp(r)$ and $\partic(r)$ from $H_r$
    \State append $(r,\ell(r),\supp(r),\partic(r),\certadm(r),\kappares(r),\resadm(r))$ to bucket $D[v]$
\EndFor
\State initialize an empty atom-record list $\mathcal B$
\For{each bucket $D[v]$}
    \State $A\gets\{r:(r,\cdot)\in D[v]\}$
    \If{any of $\supp,\partic,\certadm,\kappares,\resadm$ takes two values on $A$}
        \State \Return ``predicate not $\Fadm$-measurable''
    \EndIf
    \State create an atom record containing $A$, $\labels(A)=\{\ell(r):r\in A\}$, $\plabels(A)=\{\pi(\ell(r)):r\in A\}$, and the common predicate values on $A$
    \State append the record to $\mathcal B$
\EndFor
\State \Return $\mathcal B$
\end{algorithmic}
\end{algorithm}

\begin{algorithm}[t]
\caption{\textsc{QuotientAdmission}}
\label{alg:quotient-admission}
\begin{algorithmic}[1]
\Require atom records $\mathcal B$ produced by \textsc{BuildAtoms}
\Ensure quotient admission map $F:\Pe\to\Aset$
\For{each atom record $A\in\mathcal B$}
    \If{$|\plabels(A)|>1$}
        \State $F(A)\gets\Block$
    \ElsIf{$\plabels(A)=\{\Cset\}$ and $\certadm(A)=1$ and $\supp(A)=1$ and $\partic(A)=1$}
        \State $F(A)\gets\Cset$
    \ElsIf{$\plabels(A)=\{\Rls\}$ and $\resadm(A)=1$ and $\kappares(A)=1$}
        \State $F(A)\gets\Rls$
    \Else
        \State $F(A)\gets\Lset$
    \EndIf
\EndFor
\State \Return $F$
\end{algorithmic}
\end{algorithm}

\begin{theorem}[Correctness and maximality of quotient admission]\label{thm:correct-maximal}
Assume \textnormal{\textsc{BuildAtoms}} returns atom records rather than a measurability report. Then \textnormal{\textsc{QuotientAdmission}} returns a feasible quotient admission map. Moreover, it is maximal among all evaluators that use no hidden refinement and obey the same projected-label, certificate, residual, and conservative-projection guards.
\end{theorem}

\begin{proof}
The output is atom-level by construction: \textsc{QuotientAdmission} assigns one value to each atom record. If an atom contains more than one projected label, the first branch blocks it, so projected exact-twin conflicts are not admitted. If the algorithm outputs $\Cset$, then the certificate branch was taken, which requires $\plabels(A)=\{\Cset\}$, $\certadm(A)=1$, $\supp(A)=1$, and $\partic(A)=1$. If it outputs $\Rls$, then the residual branch was taken, which requires $\plabels(A)=\{\Rls\}$, $\resadm(A)=1$, and $\kappares(A)=1$. Atoms whose projected label set is $\{\Lset\}$ fail both positive branches and are assigned $\Lset$; this includes atoms whose semantic labels are contained in $\{\Dset,\Lset\}$. Because \textsc{BuildAtoms} verifies atom-level predicate values, no branch uses information that is unavailable in $\Fadm$.

The algorithm never admits a residual decision from residual magnitude alone. Indeed, residual magnitude is not an input to the residual branch; the branch requires the atom-level predicates $\resadm$ and $\kappares$.

For maximality, let $G:\Pe\to\Aset$ be any feasible evaluator obeying the same guards. On any atom with nonconstant projected labels, feasibility forces $G(A)=\Block$, and the algorithm also outputs $\Block$. On an atom with projected label $\Cset$ and satisfied certificate guard, any feasible map may output $\Cset$, and the algorithm does so. On an atom with projected label $\Rls$ and satisfied residual guard, any feasible map may output $\Rls$, and the algorithm does so. On every other nonblocked atom, the corresponding positive guard is absent or the projected label is $\Lset$, so no guard-respecting evaluator can assign a positive decision that the algorithm omits. Hence the algorithm is maximal with respect to positive admissions under the stated feasibility constraints.
\end{proof}

\begin{theorem}[Complexity]\label{thm:complexity}
Let $n=|\Rset|$, let
\[
B=\sum_{r\in\Rset}|e(r)|,
\qquad
I=\sum_{r\in\Rset}\sum_{W\in\mathcal W_r}|W|,
\]
where $|e(r)|$ is the explicit encoding length of the evidence vector and $I$ is the total hyperedge incidence size. If evidence vectors and hyperedges are explicitly encoded, then \textnormal{\textsc{BuildAtoms}} followed by \textnormal{\textsc{QuotientAdmission}} runs in expected $O(B+I+n)$ time and space using hashing. In the deterministic comparison model where lexicographic sorting of encoded vectors has total key-inspection cost $O(B+n\log n)$, the running time is $O(B+I+n\log n)$ and the space usage is $O(B+I+n)$. Without that key-linear sorting assumption, the deterministic fallback bound is $O(I+n\log n\cdot L)$, where $L=\max_{r\in\Rset}|e(r)|$, plus the linear cost of storing the encoded evidence.
\end{theorem}

\begin{proof}
Computing support and candidate participation scans each hyperedge incidence once, contributing $O(I)$. Reading and hashing the evidence vectors contributes $O(B)$ expected time and stores $O(B+n)$ encoded data and row references. Bucket aggregation scans each row and each stored predicate value a constant number of times, contributing $O(n)$. \textsc{QuotientAdmission} scans the atom records once, which is $O(n)$ in the worst case because there are at most $n$ atoms. Thus the expected hashing bound is $O(B+I+n)$ time and space.

For the deterministic bound, sort the explicitly encoded evidence vectors lexicographically and group equal consecutive vectors. Under a key-linear comparison model, the total number of inspected evidence symbols over the sorting and grouping phase is bounded by $O(B+n\log n)$; this holds, for example, when encoded vectors are interned, radix-sorted over a finite alphabet, or compared using longest-common-prefix accounting. The remaining scans contribute $O(I+n)$, giving $O(B+I+n\log n)$ time and $O(B+I+n)$ space. If comparisons are charged by worst-case key length without such accounting, sorting contributes $O(n\log n\cdot L)$ for $L=\max_r |e(r)|$, yielding the stated deterministic fallback bound.
\end{proof}

\section{Indistinguishability Lower Bound}\label{sec:lower-bound}
The next theorem isolates a lower bound for evaluators that observe only residual magnitudes. It is a finite indistinguishability argument: collapsed atoms cannot be separated by an oracle whose observation is constant on their union.

\begin{definition}[Magnitude-only evaluator]
Let $\resmag:\Rset\to M$ be a residual-magnitude map into a finite or countable set $M$. A magnitude-only evaluator is a map $g\circ\resmag:\Rset\to\Aset$ for some $g:M\to\Aset$.
\end{definition}

\begin{theorem}[Magnitude-only oracle lower bound]\label{thm:magnitude-lower}
Suppose there exist atoms $A,B\in\Pe$ and a value $m_0\in M$ such that $\resmag(r)=m_0$ for every $r\in A\cup B$. Suppose also that $A$ has residual label and satisfies the residual guard,
\[
\plabels(A)=\{\Rls\},\qquad \resadm(A)=1,\qquad \kappares(A)=1,
\]
while $B$ does not satisfy the residual guard or does not have projected residual label. Then no magnitude-only evaluator can agree with the maximal quotient admission map on $A\cup B$.
\end{theorem}

\begin{proof}
Every magnitude-only evaluator assigns one value $g(m_0)$ to all rows in $A\cup B$. The maximal quotient admission map assigns $\Rls$ on $A$ by Theorem~\ref{thm:correct-maximal}. It does not assign $\Rls$ on $B$ because the projected residual label or residual guard is absent there. Therefore no single value $g(m_0)$ can agree with the quotient map on both atoms.
\end{proof}

\begin{corollary}[Support compatibility cannot be recovered after collapse]
If residual magnitudes collapse two atoms that differ in support compatibility and residual admission, then residual admission is not identifiable from residual magnitude alone.
\end{corollary}

\begin{proof}
This is Theorem~\ref{thm:magnitude-lower} with the difference between the two atoms witnessed by $\kappares$ and $\resadm$.
\end{proof}

\FloatBarrier
\section{Finite Example}\label{sec:example}
Table~\ref{tab:finite-example} gives a finite instance with five atoms. Its evidence vector is $(\resmag,\supp,\partic,\certadm,\kappares,\resadm)$. The example illustrates the obstruction caused by projected-label conflicts and the distinction between semantic labels and quotient decisions.

\begin{table}[H]
\centering
\caption{Finite example of quotient admission.}
\label{tab:finite-example}
\footnotesize
\setlength{\tabcolsep}{3pt}
\renewcommand{\arraystretch}{1.08}
\begin{tabular}{@{}L{0.10\linewidth}L{0.29\linewidth}L{0.10\linewidth}L{0.20\linewidth}L{0.18\linewidth}@{}}
\toprule
Rows & Evidence vector & Atom & Semantic label & Quotient decision \\
\midrule
$a$ & $(0.08,1,1,1,0,0)$ & $A_1$ & $\Cset$ & $\Cset$ \\
$b$ & $(0.24,1,0,0,1,1)$ & $A_2$ & $\Rls$ & $\Rls$ \\
$c$ & $(0.31,1,1,1,1,1)$ & $A_3$ & $\Dset$ & $\Lset$ \\
$d,e$ & $(0.19,0,0,0,0,0)$ & $A_4$ & $\Dset,\Lset$ & $\Lset$ \\
$u$ & $(0.11,1,1,1,1,1)$ & $A_5$ & $\Cset$ & $\Block$ \\
$v$ & $(0.11,1,1,1,1,1)$ & $A_5$ & $\Rls$ & $\Block$ \\
\bottomrule
\end{tabular}
\end{table}

\textsc{BuildAtoms} groups rows with equal evidence vectors, so $d$ and $e$ form one atom and $u$ and $v$ form one atom. The atom $\{d,e\}$ maps to $\Lset$ because $\pi(\Dset)=\pi(\Lset)=\Lset$, so its projected label is constant even though its semantic labels differ. The atom $\{u,v\}$ is blocked because $\pi(\Cset)=\Cset$ and $\pi(\Rls)=\Rls$, giving a projected-label conflict that cannot be resolved without refining the evidence partition.

Rows $a$, $b$, and $c$ give nonconflicting atoms. The first atom satisfies the certificate label and guard, the second satisfies the residual label and guard, and the third carries dual information that projects to low-confidence.

\section{Conclusion}
Quotient admission is a finite combinatorial problem: an evidence map induces atoms, identifiable classes are unions of atoms, witness-support hypergraphs impose certificate guards, residual admission requires atom-level compatibility, and projected exact-twin conflicts are blocked. The resulting algorithms compute the maximal guard-respecting atom-level decision map under the no-hidden-refinement rule, with linear expected time under hashing and a model-explicit deterministic sorting bound. Magnitude-only evaluators cannot recover residual decisions after they collapse atoms whose support compatibility differs.

\bibliographystyle{plain}
\bibliography{refs}

\end{document}